\documentclass[11pt,english]{article} 
\usepackage[T1]{fontenc}
\usepackage[latin9]{inputenc}
\usepackage{hyperref}
\usepackage{verbatim}
\usepackage{amsfonts}
\usepackage{amsmath,amssymb,amsthm}
\usepackage{multirow}
\usepackage{url}

\usepackage{amsfonts}
\usepackage{courier}
\usepackage{fullpage}

\usepackage{epsfig}
\usepackage[usenames]{color}

\usepackage{bm}

\usepackage{times}

\newtheorem{theorem}{Theorem}[section]
\newtheorem{lemma}[theorem]{Lemma}
\newtheorem{proposition}[theorem]{Proposition}

\newtheorem{fact}[theorem]{Fact}
\newtheorem{corollary}[theorem]{Corollary}

\newtheorem{claim}[theorem]{Claim}

\newtheorem{definition}[theorem]{Definition}

\newenvironment{sketch}{\noindent{\bf Proof [Sketch]:~~}}{\(\qed\)}

\renewcommand{\paragraph}[1]{\noindent {\bf #1}}
\newcommand{\N}{{\mathbb N}}

\newcommand{\calA}{{\cal A}}

\newcommand{\polylog}[1]{\mathrm{polylog}{#1}}
\newcommand{\poly}[1]{\mathrm{poly}(#1)}
\newcommand{\SLC}{SLC}

\newcommand{\mcA}{{\mathcal A}}

\newcommand{\true}{\mathbf{true}}
\newcommand{\false}{\mathbf{false}}

\newcommand{\ignore}[1]{}

\makeatother

\usepackage{babel}
\makeatother

\begin{document}

\title{Fast Local Computation Algorithms}
\author{
Ronitt Rubinfeld\thanks{CSAIL, MIT, Cambridge MA 02139 and
School of Computer Science, Tel Aviv University.
 E-mail: {\tt  ronitt@csail.mit.edu}. 
 Supported by NSF grants 0732334 and 0728645,
 Marie Curie Reintegration grant PIRG03-GA-2008-231077
 and the Israel Science Foundation grant nos. 1147/09 and  1675/09.}
\and
Gil Tamir\thanks{School of Computer Science, Tel Aviv University.
 E-mail: {\tt  giltamir@gmail.com}. Supported by
 Israel Science Foundation grant no. 1147/09.}
\and
Shai Vardi
\thanks{School of Computer Science, Tel Aviv University.
 E-mail: {\tt  shaivardi@gmail.com}. Supported by
 Israel Science Foundation grant no. 1147/09. }
\and 
Ning Xie
\thanks{CSAIL, MIT, Cambridge MA 02139. 
 E-mail: {\tt  ningxie@csail.mit.edu}. 
 Part of the work was done while visiting Tel Aviv University. 
 Research supported in part by NSF Award CCR-0728645.}
}
\date{}

\setcounter{page}{0}
\maketitle

\begin{abstract}

For input $x$, let $F(x)$ denote the set of outputs
that are the ``legal'' answers for a computational
problem $F$.   Suppose $x$ and members of $F(x)$
are so large that there is not time to read them in their entirety.
We propose a model of {\em local computation algorithms}
which for a given input $x$, support queries by a user 
to values of specified locations $y_i$ in a legal output $y \in F(x)$.
When more than one legal output $y$ exists for a given $x$,
the local computation algorithm should output in a way
that is consistent with at least one such $y$.
Local computation algorithms are intended to distill the
common features of several concepts that have appeared in 
various algorithmic subfields, including local distributed
computation, local algorithms, locally decodable codes, and
local reconstruction.

We develop a technique, based on known constructions of small sample spaces of $k$-wise independent
random variables and Beck's analysis in
his algorithmic approach to the Lov{\'{a}}sz Local Lemma,
which under certain conditions can be applied 
to construct local computation algorithms that run in {\em polylogarithmic} time and space. 
We apply this technique to maximal independent set computations,
scheduling radio network broadcasts, hypergraph coloring
and satisfying $k$-SAT formulas.
\end{abstract}

\newpage

\section{Introduction}\label{Sec:intro}
Classical models of algorithmic analysis assume that 
the algorithm reads an input,
performs a computation and writes the answer to an output tape.   
On massive data sets, 
such computations may not be feasible, as
both the input and output may be too 
large for a single processor to process.  
Approaches to such situations range from proposing alternative
models of computation, such as  
parallel and distributed computation, to requiring that 
the computations achieve only approximate answers or other weaker guarantees,
as in sublinear time and space algorithms.    

In this work, we consider the scenario in which only specific
parts of the output $y = (y_1,\ldots,y_m)$ are 
needed at any point in time.    For example, if $y$ is a description
of a maximal independent set (MIS) of a graph such that $y_i=1$ if
$i$ is in the MIS and $y_i=0$ otherwise, then it may be 
sufficient to determine $y_{i_1}, \ldots, y_{i_k}$ for a small
number of $k$ vertices.
To this end, we define
a {\em local computation algorithm}, which  supports queries
of the form ``$y_{i}=?$''
such that after each query by the user to a specified location $i$, the
local computation algorithm is able to
quickly outputs $y_i$. 
For problems that allow for more than one
possible output for an input $x$,  as is often the case in combinatorial
search problems,
the local computation algorithm  must answer in such a way that 
its current answer to the query is 
consistent with any past or future answers
(in particular, there must always be at least one allowable output that is
consistent with the answers of the local computation algorithm ).
For a given problem, the hope is that the complexity of
a  local computation  algorithm  is
proportional to the amount of the solution
that is requested by the user.   
Local computation algorithms  are especially adapted to computations of  good
solutions of combinatorial search problems (cf. \cite{PR07,MR06,NO08}).   

Local computation algorithms are a formalization that describes concepts that 
are ubiquitous in the literature, and arise in various
algorithmic subfields, including local distributed
computation, local algorithms, locally decodable codes and
local reconstruction models.   
We subsequently describe the relation between these models in more detail,   
but for now, it suffices it to say that the aforementioned models are diverse 
in the computations that they apply to
(for example, whether they apply to function computations,
search problems or approximation
problems), the access to data that is allowed (for example, 
distributed computation models and other local algorithms 
models often specify that the input is a graph in adjacency list format),
and the required running time bounds (whether the computation time is required
to be independent of the problem size, or whether it is allowed
to have computation time that depends on the problem
size but has sublinear complexity).
The model of local computation algorithms described here 
attempts to distill the essential common features of
the aforementioned concepts into a unified paradigm.

After formalizing our model,
we develop a technique that is specifically suitable 
for constructing polylogarithmic time local computation algorithms.
This technique is based on Beck's analysis in his
algorithmic approach to the Lov{\'{a}}sz Local Lemma (LLL) ~\cite{Bec91},
and uses the underlying locality of the problem to construct a solution.
All of our constructions must process the data somewhat so
that Beck's analysis applies, and we use two main methods
of performing this processing.

For maximal independent set computations and
scheduling radio network broadcasts,
we use a reduction  of Parnas and
Ron \cite{PR07},  which 
 shows that, given a distributed network in which
the underlying graph has degree bounded by a constant $D$,
and a distributed algorithm (which may be randomized
and only guaranteed to 
output an approximate answer) using at most $t$ rounds,
the output of any specified node of the distributed network
can be simulated by a single processor with query access to
the input with  $O(D^{t+1})$ 
queries.   
For our first phase, we apply this reduction to a constant number
of rounds of Luby's maximal independent set distributed algorithm \cite{Lub86}
in order to find a partial solution. 
After a large independent set is found,
in the second phase, we use the techniques of Beck \cite{Bec91} to show 

For hypergraph coloring and $k$-SAT,   
we show that Alon's parallel algorithm,  
which is guaranteed to find solutions by the Lov{\'{a}}sz Local Lemma~\cite{Alo91},
can  be modified to run in polylogarithmic sequential
time to answer queries regarding
the coloring of a given node, or the setting of
a given variable.
For most of the queries,
this is done by simulating the work of a processor and its neighbors
within a small constant radius.    
For the remainder of the queries,
we use Beck's analysis to show that the queries can be solved via the
brute force algorithm on very small subproblems.

Note that parallel $O(\log{n})$ time algorithms do not
directly yield local algorithms under the reduction
of Parnas and Ron \cite{PR07}. Thus,
we do not know how to apply our techniques to all problems
with fast parallel algorithms, including certain problems
in the work of Alon
whose solutions are guaranteed by the
Lov{\'{a}}sz Local Lemma and which have fast parallel algorithms
~\cite{Alo91}.
Recently Moser and Tardos~\cite{Mos09, MT10} gave, 
under some slight restrictions,  
parallel algorithms which find solutions of 
all problems  for which the existence of
a solution  is guaranteed by the 
Lov{\'{a}}sz Local Lemma~\cite{Bec91,Alo91}.
It has not yet been resolved whether one can
construct local computation algorithms based on these more powerful algorithms.

\subsection{Related work}
Our focus on local
computation algorithms is inspired by many existing works, which
explicitly or implicitly construct such procedures.
These results occur in a number of varied settings,
including distributed systems, coding theory,
and sublinear time algorithms.   

Local computation algorithms
are a generalization of {\em local algorithms},  which for
graph theoretic problems represented in adjacency
list format, produce a solution
by adaptively examining only a constant sized portion of the
input graph near a specified vertex.  
Such algorithms have received much attention in the
distributed computing literature   
under a somewhat different model, in which the number of rounds
is bounded to constant time, but the computation is performed
by all of the processors in the distributed network \cite{NS95, MNS95}.
Naor and Stockmeyer \cite{NS95} and Mayer, Naor
and Stockmeyer \cite{MNS95}
investigate the question of what can be computed  under these 
constraints,
and show that there are nontrivial problems with such algorithms.
Several more recent works 
investigate local algorithms for various problems, including
coloring, maximal independent set, dominating set 
(some examples are in \cite{KMW04,KMW06,KMN+05, KW06, KM07,K09, SW08, BE09, BE10, SW10}).
Although all of these algorithms are distributed algorithms, those
that use constant rounds 
yield (sequential) local computation algorithms via
the previously mentioned reduction of Parnas and 
Ron \cite{PR07}.\footnote{Note that the parallel
algorithms of Luby \cite{Lub86}, Alon 
\cite{Alo91} and Moser and Tardos \cite{MT10}
do not automatically yield local algorithms via this transformation
since their parallel running times are $t= \Omega(\log{n})$.}

There has been much recent interest among the sublinear
time algorithms community in devising local algorithms
for problems on constant degree graphs and other
sparse optimization problems.
The goals of these algorithms have been to approximate
quantities such as the optimal vertex cover, maximal matching,
maximum matching, dominating set, sparse set cover, 
sparse packing and cover problems
\cite{PR07, KMW06, MR06,NO08,HKN+09,YYI09}.
One feature of these algorithms is that they show how
to construct an oracle which for each vertex returns
whether it part of the solution whose size is being
approximated -- for example, whether it 
is in the vertex cover or maximal matching.   Their
results show that this oracle can be implemented in
time independent of the size of the graph (depending
only on the maximum degree and the approximation parameter).
However, because their
goal is only to compute an approximation of their quantity,
they can afford to err on a small fraction of their local computations.
Thus, their oracle implementations give local computation algorithms for 
finding relaxed solutions to the the optimization problems that
they are designed for. For example, constructing a local 
computation algorithm using the oracle designed 
for estimating the size of a maximal independent set 
in \cite{MR06} yields a large independent set, 
but not necessarily a maximal independent set.
Constructing a local computation algorithm using the oracle
designed for  estimating the size of
the vertex cover \cite{PR07,MR06,NO08}
yields a vertex cover whose size is 
guaranteed to be only slightly
larger than what is given by the $2$-approximate 
algorithm being emulated -- namely, by a multiplicative
factor of at most $2+\delta$ (for any $\delta>0$).

Recently, local algorithms have been demonstrated to
be applicable for computations on the web graph.
In \cite{JW03,Ber06,SBC+06,ACL06,ABC+08}, local algorithms are given which,
for a given vertex $v$ in the web graph, computes 
an approximation to $v$'s personalized PageRank vector
and computes the vertices that contribute significantly 
to $v$'s PageRank.  In these algorithms, evaluations
are made only to the nearby neighborhood of $v$, so
that the running
time depends on the accuracy parameters input to the
algorithm, but there
is no running time dependence on the size of the web-graph.
Local graph partitioning algorithms have been
presented in \cite{ST04,ACL06} which find subsets of vertices
whose internal connections are significantly richer
than their external connections.  The running time of
these algorithms depends on the size of the cluster
that is output, which can be much smaller than the size of
the entire graph.

Though most of the previous examples are for sparse
graphs or other problems which have some sort of sparsity, local computation
algorithms have also been provided for problems on dense graphs.
The property testing algorithms of
\cite{GGR98} use a small sample of the vertices
(a type of a core-set) to define a good
graph coloring or partition of a dense graph.  This approach
yields local computation algorithms for finding a large partition of the graph
and a coloring of the vertices which has relatively few
edge violations.

The applicability of local computation algorithms is
not restricted to combinatorial problems.   One
algebraic paradigm that has local computation algorithms
is that of {\em locally decodable
codes} \cite{KT00},
described by the following scenario:  Suppose $m$
is a string with encoding $y=E(m)$.  On input $x$, which
is close in Hamming distance to $y$, the goal of locally
decodable coding algorithms is to provide quick access
to the requested bits of $m$.  
More generally, the {\em reconstruction} models
described in \cite{ACC+08, CS06a,SS10} describe scenarios
where a string that has a certain property, such as monotonicity, is assumed
to be corrupted at a relatively small number of
locations.  Let $P$ be the set of strings that
have the property.   The reconstruction algorithm gets as
input a string $x$ which is close (in $L_1$ norm),
to some string $y$ in $P$.   For various types of
properties $P$, the above works construct
algorithms which give fast query access to locations in $y$.

\subsection{Followup work}
In a recently work, Alon et al.~\cite{ARVX11}
further show that the local computation algorithms in this work can
be modified to not only run in polylogarithmic time but also in polylogarithmic space.
Basically, they show that all the problems studied in this paper can be solved in 
a unified local algorithmic framework as follows: 
first the algorithm generates some random bits on a read-only random tape of 
size at most polylogarithmic in the input, 
which may be thought as the ``core'' of a solution $y \in F(x)$; 
then each bit of the solution $y$ can be computed \emph{deterministically}
by first querying a few bits on the random tape and 
then performing some fast computations on these random bits.  
The main technical tools used in~\cite{ARVX11} are pseudorandomness,
the random ordering idea of~\cite{NO08} and the theory of branching processes.

\subsection{Organization}
The rest of the paper is organized as follows.
In Section~\ref{Sec:model} we present our computation model.
Some preliminaries and notations that we use throughout the paper
appear in Section~\ref{Sec:prelim}.
We then give local computation
algorithms for the maximal independent set problem
and the radio network broadcast scheduling problem 
in Section~\ref{Sec:maxis} and Section~\ref{Sec:radio}, respectively.
In Section~\ref{Sec:hypergraph} we show how
to use the parallel algorithmic version of the Lov{\'{a}}sz
Local Lemma to give local computation algorithms
for finding the coloring of nodes in a  hypergraph.
Finally, in Section~\ref{Sec:kcnf},
we show how to find settings of variables according
to a satisfying assignment of a $k$-CNF formula.

\section{Local Computation Algorithms: the model}\label{Sec:model}

We present our model of {\em local 
computation algorithms} for sequential computations of search problems,
although computations of arbitrary functions and optimization
functions also fall within our framework.  

\subsection{Model Definition}
We write $n=|x|$ to denote the length of the input.
\begin{definition}
For input $x$, let $F(x) = \{y~|~ y {\rm ~is~ a~} {\rm valid~solution~ 
}{\rm for~
input~} x\}$.
The {\em search problem} is to find any $y \in F(x)$.   
\end{definition}

In this paper, the description 
of both $x$ and $y$ are assumed to be very large.

\begin{definition}[$(t, s, \delta)$-local algorithms]
Let $x$ and $F(x)$ be defined as above.
A {\em $(t(n), s(n), \delta(n))$-local computation algorithm} $\mcA$ 
is a (randomized) algorithm which 
implements query access to an arbitrary $y \in F(x)$ and satisfies the following:
$\mcA$ gets a sequence of queries
$i_1,\ldots,i_q$ for any $q>0$
and 
after each query $i_j$ it must produce an output
$y_{i_j}$ 
satisfying that 
the outputs $y_{i_1},\ldots, y_{i_q}$ are substrings of
some $y \in F(x)$.
The probability of success over all $q$ queries must be
at least $1-\delta(n)$.
$\mcA$ has access to a random tape and local computation memory on which
it can perform current computations as well as store and retrieve information
from previous computations.
We assume that the input $x$, the local
computation tape and any
random bits used are all presented in the RAM word 
model,
i.e., $\mcA$ is given the ability to access
a word of any of these in one step.
The running time of $\mcA$ on any query is at most $t(n)$,
which is sublinear in $n$,
and the size of the local computation memory of $\mcA$ is at most $s(n)$.
Unless stated otherwise, we always assume that the error parameter $\delta(n)$ is
at most some constant, say, $1/3$.
We say that  $\mcA$ is a {\em strongly local computation algorithm}
if both $t(n)$ and $s(n)$ are upper bounded by $\log^{c}n$ for some constant $c$.  
\end{definition}

\begin{definition}
Let $\SLC$ be the class of problems that have strongly local
computation algorithms.
\end{definition}

\noindent
Note that when $|F(x)|>1$, the $y$ according to which $\mcA$ outputs 
may depend on the previous queries
to $\mcA$ as well as any random bits available to $\mcA$.
Also, we implicitly assume that the size of the output
$y$ is upper-bounded by some polynomial in $|x|$.  
The definition of local-computation algorithms
rules out the possibility that the algorithms accomplish their computation by
first computing the entire output.
Analogous definitions can be made for a
bit model. 
In principle, the model
applies to general computations, including function computations,
search problems and optimization problems of any type of object,
and in particular, the input is not required by the model
to be in a specific input format.

The model presented here is intended be more general, and 
thus differs from other local computation models in
the following ways.   First, 
queries and processing time have the same cost.
Second, the focus is on problems with slightly looser running
time bound requirements -- polylogarithmic
dependence on the length of the input is desirable,
but sublinear time in the length of the input is often
nontrivial and can be acceptable.  
Third, the model places no restriction on 
the ability of the algorithm to access the input, as is
the case in the distributed setting where the algorithm
may only query nodes in its neighborhood (although 
such restrictions may be implied by the representation
of the input).   
As such, the model may be 
less appropriate for certain distributed algorithms
applications.

\begin{definition}[Query oblivious]
We say an LCA $\mcA$ is \emph{query order oblivious} (\emph{query oblivious} for short)
if the outputs of $\mcA$ do not depend on the order of the queries but
depend only on the input and the random bits generated by the random tape of $\mcA$.
\end{definition}

\begin{definition}[Parallelizable]
We say an LCA $\mcA$ is \emph{parallelizable} if $\mcA$ 
supports parallel queries.
\end{definition}

We remark that not all local algorithms in this paper are
query oblivious or easily parallelizable. 
However, this is remedied in~\cite{ARVX11}.

\subsection{Relationship with other distributed and parallel models}
A question that immediately arises is to characterize
the problems to which the local-computation algorithm model applies.
In this subsection, we note the relationship between
problems solvable with local computation algorithms 
and those solvable with fast parallel or distributed algorithms.

From the work of \cite{PR07} it follows that problems
computable by fast distributed algorithms also have
local computation algorithms.

\begin{fact}[\cite{PR07}] \label{fact:PR07}  
If $F$ is computable in $t(n)$ rounds on 
a distributed network in which the processor interconnection
graph has bounded degree $d$, then $F$ has a $d^{t(n)}$-local
computation algorithm.
\end{fact}

Parnas and Ron~\cite{PR07} show this fact
by observing that for any vertex $v$, 
if we run a distributed algorithm $\calA$ on the
subgraph $G_{k,v}$ (the vertices of distance at most $k$ from $v$), 
then it makes the same decision about vertex $v$ as it would if we would
 run $D$ for $k$ rounds on the whole graph $G$. 
They then give a reduction from randomized distributed algorithms to 
sublinear algorithms based on this observation.

Similar relationships hold in other distributed and parallel models,
in particular, for problems computable by low depth bounded fan-in
circuits.

\begin{fact} If $F$ is computable by a circuit family of depth $t(n)$
and fan-in bounded by  $d(n)$, then $F$ has a $d(n)^{t(n)}$-local
computation algorithm.
\end{fact}

\begin{corollary}
$NC^0 \subseteq \SLC$.
\end{corollary}

In this paper we show solutions to several 
problems $NC^1$ via local computation algorithms.   
However, this is not possible in general as:\\

\begin{proposition}
$NC^1 \nsubseteq \SLC$. 
\end{proposition}

\begin{proof}
Consider the problem {\em n-XOR},  the $XOR$ of $n$ inputs. 
This problem is in $NC^1$. 
However, no sublinear time algorithm can solve {\em n-XOR} because it 
is necessary to read all $n$ inputs.
\end{proof}

In this paper, we give techniques which allow one to construct
local computation algorithms based on algorithms for finding
certain combinatorial structures whose
existence is guaranteed by constructive
proofs of the LLL in~\cite{Bec91,Alo91}. It seems that our techniques do not extend to all such problems. 
An example of such a problem is {\em Even cycles 
in a balanced digraph}: find an even cycle in a digraph whose 
maximum in-degree is not much greater that the minimum out-degree.
Alon ~\cite{Alo91} shows that, under certain restriction on the input parameters,
the problem is in $NC^1$.
The local analogue of this question is to determine 
whether a given edge (or vertex) is part of an even cycle in such a graph.
It is not known how to solve this quickly.

\subsection{Locality-preserving reductions}

In order to understand better which problems can be solved
locally, we define
{\em locality-preserving reductions},  which capture the idea
that if problem $B$ is locally computable, and problem $A$ has a 
locality-preserving reduction to $B$ then $A$ is also locally computable.

\begin{definition}
We say that $A$ is {\em $(t'(n), s'(n))$-locality-preserving reducible} to $B$  
via reduction $H: \Sigma^{*} \to \Gamma^{*}$, where $\Sigma$ and $\Gamma$ are
the alphabets of $A$ and $B$ respectively, if
$H$ satisfies:
\begin{enumerate}
\item  $x \in A  \iff  H(x) \in B$.
\item  $H$ is $(t'(n), s'(n), 0)$-locally computable;  that is, every word of $H(x)$ can be
       computed by querying at most $t(n)$ words of $x$.
\end{enumerate}
\end{definition}

\begin{theorem}
If $A$ is $(t'(n), s'(n), 0)$-locality-preserving reducible 
to $B$ and $B$ is $(t(n), s(n), \delta(n))$-locally computable, 
then $A$ is $(t(n) \cdot t'(n), s(n)+s'(n), \delta(n))$-locally computable.
\end{theorem}

\begin{proof}
As $A$ is $(t'(n), s'(n), 0)$-locality-preserving
reducible to $B$, 
to determine whether $x \in A$, 
it suffices to determine if $H(x) \in B$. Each word of $H(x)$ can 
be computed in time $t'(n)$ and using space $s'(n)$, and we need to access at most $t(n)$ 
such words to determine whether $H(x) \in B$. 
Note that we can reuse the space for computing $H(x)$.
\end{proof}

\section{Preliminaries}
\label{Sec:prelim}

Unless stated otherwise, all logarithms in this paper are to the base $2$.
Let $\N=\{0,1,\ldots\}$ denote the set of natural numbers.
Let $n\geq 1$ be a natural number. We use $[n]$ to denote the set $\{1,\ldots, n\}$.

Unless stated otherwise, all graphs are undirected. Let $G=(V, E)$ be a graph.
The \emph{distance} between two vertices $u$ and $v$ in $V(G)$, denoted by
$d_{G}(u, v)$, is the length of a shortest path between the two vertices.
We write $N_{G}(v)=\{u\in V(G): (u,v)\in E(G)\}$ to denote
the neighboring vertices of $v$.   Furthermore, 
let $N^{+}_{G}(v)=N(v)\cup \{v\}$.
Let $d_{G}(v)$ denote the degree of a vertex $v$. 
Whenever there is no risk of confusion, we omit the subscript $G$ from $d_{G}(u,v)$, $d_{G}(v)$
and $N_{G}(v)$. \\

The celebrated Lov{\'{a}}sz Local Lemma plays an important role in our results. 
We will use the simple symmetric version of the lemma.
\begin{lemma}[Lov{\'{a}}sz Local Lemma~\cite{EL75}]\label{lemma:LLL}
Let $A_{1}, A_{2}, \ldots, A_{n}$ be events in an arbitrary probability space.
Suppose that the probability of each of these $n$ events is at most $p$, 
and suppose that each event $A_{i}$ is mutually independent of all
but at most $d$ of other events $A_{j}$. If $ep(d + 1) \leq 1$, 
then with positive probability none of the events $A_{i}$ holds, i.e.,
\[
\Pr[\cap_{i=1}^{n}\bar{A}_{i}]>0.
\]
\end{lemma}

Several of our proofs use the following graph theoretic structure:
\begin{definition}[\cite{Bec91}]
Let $G=(V, E)$ be an undirected graph.
Define $W \subseteq V(G)$ to be a {\em $3$-tree}
if the pairwise distances of all vertices in $W$ are each at least $3$
and the graph $G^{*} = (W, E^*)$ is connected, 
where $E^*$ is the set of edges between
each pair of vertices whose distance is exactly $3$ in $G$.  
\end{definition}

\section{Maximal Independent Set}\label{Sec:maxis}
An independent set (IS) of a graph $G$ is a subset of vertices such that no two vertices 
are adjacent. An independent set is called a \emph{maximal independent set}
(MIS) if it is not properly contained in any other IS. 
It is well-known that a 
sequential greedy algorithm finds an MIS $S$ in linear time:
Order the vertices in $G$ as $1,2,\cdots, n$ and initialize $S$ to the empty set;
for $i=1$ to $n$, if vertex $i$ is not adjacent to any vertex in $S$, add $i$ to $S$.
The MIS obtained by this algorithm is call the \emph{lexicographically first maximal independent set} (LFMIS).
Cook~\cite{Coo85} showed that deciding if vertex $n$ is in the LFMIS is $P$-complete with
respect to logspace reducibility.
On the other hand,
fast randomized
parallel algorithms for MIS were discovered in 
1980's~\cite{KW85, Lub86, ABI86}. 
The best known distributed algorithm for MIS runs
in $O(\log^{*}n)$ rounds with a word-size of $O(\log n)$~\cite{GPS88}.
By Fact~\ref{fact:PR07}, this implies
a $d^{O(\log^{*}n \cdot \log n)}$ local computation algorithm.
In this section, we give a query oblivious and parallelizable 
local computation algorithm for MIS
based on Luby's algorithm as well as the techniques of Beck~\cite{Bec91},
and runs in time $O(d^{O(d\log d)}\cdot\log n) $.

Our local computation algorithm is partly inspired by the work of Marko and Ron~\cite{MR06}.
There they simulate the distributed algorithm for MIS of Luby~\cite{Lub86} in order
to \emph{approximate} the minimum number of edges one need to remove to make
the input graph free of some fixed graph $H$.
In addition, they show that similar algorithm can also approximate the size
of a minimum vertex cover. 
We simulate Luby's algorithm to find an \emph{exact and consistent} local solution
for the MIS problem. Moreover, the ingredient of applying Beck's idea
to run a second stage greedy algorithm on disconnected subgraphs 
seems to be new.   

\subsection{Overview of the algorithm}\label{Sec:max_is_overview}
Let $G$ be an undirected graph on $n$ vertices and with maximum degree $d$.
On input a vertex $v$, 
our algorithm decides 
whether $v$ is in a maximal independent set using two phases. 
In Phase $1$, we simulate Luby's parallel algorithm for MIS~\cite{Lub86} via
the reduction of \cite{PR07}.
That is, in each round, $v$ tries to
put itself into the IS with some small probability. 
It succeeds if none of its neighbors
also tries to do the same. We run our Phase $1$ algorithm for $O(d\log d)$ rounds.
As it turns out, after Phase $1$,
most vertices have been either added to the IS or removed from the graph due to 
one (or more) of their neighbors being in the IS. 
Our key observation is that -- following a variant of the argument of Beck~\cite{Bec91}
-- almost surely, all the connected components of the surviving vertices after Phase $1$
have size at most $O(\log n)$. This enables us to
perform the greedy algorithm for the connected component $v$ lies in.

Our main result in this section is the following.
\begin{theorem}\label{thm:MIS_main}
Let $G$ be an undirected graph with $n$ vertices and maximum degree $d$. 
Then there is a 
$(O(d^{O(d\log d)}\cdot \log{n}), O(n), 1/n)$-local computation algorithm 
which, on input a vertex $v$, 
decides if $v$ is in a maximal independent set.
Moreover, the algorithm will give a consistent MIS 
for every vertex in $G$.
\end{theorem}

\subsection{Phase 1: simulating Luby's parallel algorithm}\label{Sec:max_is_phase1}
Figure~\ref{maxis} illustrates Phase $1$ of our local computation algorithm
for Maximal Independent Set. 
Our algorithm simulates Luby's algorithm for $r=O(d\log d)$ rounds.
Every vertex $v$ will be in one of three possible states: 
\begin{itemize}
\item ``selected'' --- $v$ is in the MIS;
\item ``deleted'' --- one of $v$'s neighbors is selected and $v$ is deleted from the graph; and
\item ``$\perp$'' --- $v$ is not in either of the previous states. 
\end{itemize}
Initially, every vertex is in state ``$\perp$''.
Once a vertex becomes ``selected'' or ``deleted'' in some round, 
it remains in that state in all the subsequent rounds.

The subroutine $\mathbf{MIS}(v, i)$ returns the state of a vertex $v$ in round $i$.
In each round, if vertex $v$ is still in state ``$\perp$'', 
it ``chooses'' itself to be in the MIS with probability $1/2d$.
At the same time, all its neighboring vertices also flip random coins to decide if 
they should ``choose'' themselves.\footnote{We store all the randomness generated by 
each vertex in each round so that our answers will be consistent. However, we
generate the random bits only when the state of corresponding vertex in that round is requested.}
If $v$ is the only vertex in $N^{+}(v)$ that is chosen in that round,
we add $v$ to the MIS (``select'' $v$) and ``delete'' all the neighbors of $v$.
However, the state of $v$ in round $i$ is determined not only by the random coin tosses 
of vertices in $N^{+}(v)$ but also by these vertices' states in round $i-1$.
Therefore, to compute $\mathbf{MIS}(v, i)$, 
we need to recursively call $\mathbf{MIS}(u, i-1)$
for every $u\in N^{+}(v)$. \footnote{A subtle point in 
the subroutine $\mathbf{MIS}(v, i)$ is that when we call
$\mathbf{MIS}(v, i)$, we only check if vertex $v$ is ``selected'' or
not in round $i$. If $v$ is ``deleted'' in round $i$, 
we will not detect this until we call $\mathbf{MIS}(v, i+1)$,
which checks if some neighboring vertex of $v$ is selected in round $i$.
However, such ``delayed'' decisions will not affect our analysis of the algorithm.}
By induction, the total running time of simulating
$r$ rounds is $d^{O(r)}=d^{O(d\log d)}$.

If after Phase $1$ all vertices are either ``selected'' or ``deleted'' and
no vertex remains in ``$\perp$'' state,
then the resulting independent set is a maximal independent set.
In fact, one of the main results in~\cite{Lub86} is that this indeed is the case
if we run Phase $1$ for expected $O(\log{n})$ rounds.
Our main observation is, after simulating Luby's algorithm 
for only $O(d\log d)$ (a constant independent of the size $n$) rounds 
we are already not far from a maximal independent set. 
Specifically, if vertex $v$ returns ``$\perp$'' after Phase $1$ of the algorithm,
we call it a \emph{surviving} vertex. 
Now consider the subgraph induced on the surviving vertices. 
Following a variant of Beck's argument~\cite{Bec91}, we
show that, almost surely, no connected component of surviving vertices is larger than 
$\poly{d}\cdot\log n$.

Let $A_{v}$ be the event that vertex $v$ is a surviving vertex.
Note that event $A_v$ depends on the random 
coin tosses $v$ and $v$'s neighborhood of radius $r$
made during the first $r$ rounds, where $r=O(d\log d)$. To get rid of
the complication caused by this dependency, we consider another set of events
generated by a related random process.

\begin{figure*}
\begin{center}
\fbox{
\begin{minipage}{5in}
\small
\begin{tabbing}
\textsc{Maximal Independent Set: Phase $1$} \\
Input: a graph $G$ and a vertex $v \in V$\\
Out\=put\=:  \{``\=tru\=e'', \=``fa\=lse'', ``$\perp$''\}\\
   \> For $i$ from $1$ to $r=20d\log d$ \\
   \>\>(a) If $\mathbf{MIS}(v, i)=\text{``selected''}$ \\
     \>\>\> return ``true''\\
   \>\>(b) Else if $\mathbf{MIS}(v, i)=\text{``deleted''}$ \\
	  \>\>\> return ``false'' \\
   \>\>(c) Else \\
	\>\>\> return ``$\perp$''\\
\\

{$\mathbf{MIS}(v, i)$}\\
Input: a vertex $v \in V$ and a round number $i$\\
Out\=put\=:  \{``\=sel\=ect\=ed'', ``deleted'', ``$\perp$''\}\\
1. \> If $v$ is marked ``selected'' or ``deleted'' \\
   \>\> return ``selected'' or ``deleted'', respectively\\
2. \> For every $u$ in $N(v)$\\
   \>\> If $\mathbf{MIS}(u, i-1)=\text{``selected''}$ \\
   \>\>\> mark $v$ as `` deleted'' and return ``deleted''\\
3. \> $v$ chooses itself independently with probability $\frac{1}{2d}$\\
   \>\> If $v$ chooses itself\\
   \>\>\> (i) For every $u$ in $N(v)$\\
   \>\>\>\> If $u$ is marked ``$\perp$'', 
       $u$ chooses itself independently with probability $\frac{1}{2d}$ \\
   \>\>\> (ii) If $v$ has a chosen neighbor \\
   \>\>\>\> return ``$\perp$''\\
   \>\>\> (iii) Else\\
   \>\>\>\> mark $v$ as ``selected'' and return ``selected''\\
   \>\> Else \\
   \>\>\> return ``$\perp$''
  
\end{tabbing}
\end{minipage}
}
\end{center}
\caption{Local Computation Algorithm for MIS:  Phase $1$}
\label{maxis}
\end{figure*}

Consider a variant of our algorithm $\mathbf{MIS}$, which we call $\mathbf{MIS}_{B}$
as shown in Fig~\ref{maxis_B}.
In $\mathbf{MIS}_{B}$, every vertex $v$ has two possible states:
``picked'' and ``$\perp$''.
Initially, every vertex is in state ``$\perp$''.
Once a vertex becomes ``picked'' in some round, 
it remains in the ``picked'' state in all the subsequent rounds.
$\mathbf{MIS}_{B}$ and $\mathbf{MIS}$ are identical except that, 
in  $\mathbf{MIS}_{B}$, 
if in some round a vertex $v$ is the only vertex in $N^{+}(v)$ that is chosen in that round,
the state of $v$ becomes ``picked'',
but we do not change the states of $v$'s neighboring vertices.
In the following rounds, $v$ keeps flipping coins and tries to choose itself.
Moreover, we assume that, for any vertex $v$ and in every round, 
the randomness used in $\mathbf{MIS}$ and $\mathbf{MIS}_{B}$ are identical
as long as $v$ has not been ``selected'' or ``deleted'' in $\mathbf{MIS}$.
If $v$ is ``selected'' or ``deleted'' in $\mathbf{MIS}$,
then we flip some additional coins for $v$ in the subsequent rounds 
to run $\mathbf{MIS}_{B}$.  
We let $B_{v}$ be the event that $v$ is in state ``$\perp$'' after running 
$\mathbf{MIS}_{B}$ for $r$ rounds 
(that is, $v$ is never get picked during all $r$ rounds of $\mathbf{MIS}_{B}$).

\begin{figure*}
\begin{center}
\fbox{
\begin{minipage}{5in}
\small
\begin{tabbing}
{$\mathbf{MIS}_{B}(v, i)$}\\
Input: a vertex $v \in V$ and a round number $i$\\
Out\=put\=:  \{``\=pick\=ed'', ``$\perp$''\}\\
1. \> If $v$ is marked ``picked'' \\
   \>\> return ``picked''\\
2. \> $v$ chooses itself independently with probability $\frac{1}{2d}$\\
   \>\> If $v$ chooses itself\\
   \>\>\> (i) For every $u$ in $N(v)$\\
   \>\>\>\>  $u$ chooses itself independently with probability $\frac{1}{2d}$ \\
   \>\>\> (ii) If $v$ has a chosen neighbor \\
   \>\>\>\> return ``$\perp$''\\
   \>\>\> (iii) Else\\
   \>\>\>\> mark $v$ as ``picked'' and return ``picked''\\
   \>\> Else \\
   \>\>\> return ``$\perp$''
  
\end{tabbing}
\end{minipage}
}
\end{center}
\caption{Algorithm $\mathbf{MIS}_{B}$}
\label{maxis_B}
\end{figure*}

\begin{claim}
$A_{v} \subseteq B_{v}$ for every vertex $v$.\footnote{Strictly speaking, 
the probability spaces in which $A_{v}$ and $B_{v}$ live are different.
Here the claim holds for any fixed outcomes of all the additional random coins
$\mathbf{MIS}_{B}$ flips.}
\end{claim}
\begin{proof}
This follows from the facts that
a necessary condition for $A_{v}$ to happen is $v$ never get ``selected''
in any of the $r$ rounds and 
deleting the neighbors of $v$ from the graph can never decrease 
the probability that $v$ gets ``selected'' in any round.
Specifically, we will show that $\overline{B}_{v} \subseteq \overline{A}_{v}$.

Note that $\overline{B}_{v}=\cup_{i=1}^{r}\overline{B}^{(i)}_{v}$,
where $\overline{B}^{(i)}_{v}$ is the event that $v$ is picked 
for the first time in round $i$ in $\mathbf{MIS}_{B}$ ($v$ may get picked
again in some subsequent rounds).
Similarly, $\overline{A}_{v}=\cup_{i=1}^{r}\overline{A}^{(i)}_{v}$,
where $\overline{A}^{(i)}_{v}$ is the event 
that $v$ is selected or deleted in round $i$ in $\mathbf{MIS}$.
Hence, we can write, for every $i$, 
$\overline{A}^{(i)}_{v}=\mathrm{Sel}^{(i)}_{v} \cup \mathrm{Del}^{(i)}_{v}$,
where $\mathrm{Sel}^{(i)}_{v}$ and $\mathrm{Del}^{(i)}_{v}$
are the events that, in $\mathbf{MIS}$, 
$v$ gets selected in round $i$ and $v$ get deleted in round $i$, respectively.

We prove by induction on $i$ that 
$\cup_{j=1}^{i}\overline{B}^{(j)}_{v} \subseteq \cup_{j=1}^{i}\overline{A}^{(j)}_{v}$.
This is clearly true for $i=1$ as $\overline{B}^{(1)}_{v}=\mathrm{Sel}^{(1)}_{v}$.
Assume it holds for all smaller values of $i$. 
Consider any fixed random coin tosses of all the vertices in the graph before round $i$ such that
$v$ is not ``selected'' or ``deleted'' before round $i$.
Then by induction hypothesis, $v$ is not picked in $\mathbf{MIS}_{B}$
before round $i$ either. 
Let $N^{(i)}(v)$ be the set of neighboring nodes of $v$ that are in state ``$\perp$''
in round $i$ in algorithm $\mathbf{MIS}_{B}$. Clearly, $N^{(i)}(v) \subseteq N(v)$.

Now for the random coins tossed in round $i$, we have
\begin{align*}
 \overline{B}^{(i)}_{v} 
 &= \{\text{$v$ chooses itself in round $i$} \} \cap_{w\in N(v)} 
    \{\text{$w$ does not chooses itself in round $i$} \} \\
 &\subseteq \{\text{$v$ chooses itself in round $i$} \} \cap_{w\in N^{(i)}(v)} 
    \{\text{$w$ does not chooses itself in round $i$} \} \\
 &= \mathrm{Sel}^{(i)}_{v}.
\end{align*}
Therefore, $\overline{B}^{(i)}_{v} \subseteq \mathrm{Sel}^{(i)}_{v} \subseteq \overline{A}^{(i)}_{v}$.
This finishes the inductive step and thus
completes the proof of the claim.
\end{proof}

As a simple corollary, we immediately have
\begin{corollary}
For any vertex set $W \subset V(G)$, $\Pr[\cap_{v\in W}A_{v}] \leq \Pr[\cap_{v\in W}B_{v}]$.
\end{corollary}

A graph $H$ on the vertices $V(G)$ 
is called a \emph{dependency graph} for $\{B_{v}\}_{v\in V(G)}$ if for all $v$ the
event $B_{v}$ is mutually independent of all $B_{u}$ such that $(u,v)\notin H$.

\begin{claim}\label{claim:survived_degree}
The dependency graph $H$ has maximum degree $d^2$.
\end{claim} 
\begin{proof}
Since for every vertex $v$, $B_{v}$ depends only on the
coin tosses of $v$ and vertices in $N(v)$ in each of the $r$ rounds,
the event $B_{v}$ is independent of all $A_{u}$ such that
$d_{H}(u, v)\geq 3$. The claim follows as there are
at most $d^2$ vertices at distance $1$ or $2$ from $v$.
\end{proof}

\begin{claim}\label{claim:survived_prob}
For every $v\in V$, the probability that $B_{v}$ occurs is at most $1/8d^3$.
\end{claim}
\begin{proof}
The probability that vertex $v$ is chosen in round $i$ is $\frac{1}{2d}$. 
The probability that none of its neighbors is chosen in this round is 
$(1-\frac{1}{2d})^{d(v)} \geq (1-\frac{1}{2d})^{d} \geq 1/2$. 
Since the coin tosses of $v$ and vertices in $N(v)$ are independent,
the probability that $v$ is selected in round $i$ is at least 
$\frac{1}{2d}\cdot \frac{1}{2}=\frac{1}{4d}$.
We get that the probability that $B_{v}$ happens
is at most $(1-\frac{1}{4d})^{20d\log d} \leq \frac{1}{8d^3}$.
\end{proof}

Now we are ready to prove the main lemma for our local computation algorithm for MIS.
\begin{lemma}\label{lemma:survived_size}
After Phase $1$, with probability at least $1-1/n$, all connected components 
of the surviving vertices are of size at most $O(\poly{d} \cdot \log n)$.
\end{lemma}
\begin{proof}
Note that we may upper bound the probability that all vertices in $W$ are surviving vertices 
by the probability that all the events $\{B_v\}_{v\in W}$ happen simultaneously: 
\begin{align*}
&\quad\Pr[\text{all vertices in $W$ are surviving vertices}]\\
&=\Pr[\cap_{v\in W} A_{v}]  \\
&\leq\Pr[\cap_{v\in W} B_{v}].
\end{align*}
The rest of the proof is similar to that of Beck~\cite{Bec91}.
We bound the number of $3$-trees in $H$ (the dependency graph 
for events $\{B_{v}\}$) of $3$-trees of size $w$ as follows. 

Let $H^3$ denote the ``distance-$3$'' graph of $H$, that is,
vertices $u$ and $v$ are connected in $H^3$ if their distance in $H$ is exactly $3$.
We claim that, for any integer $w>0$, 
the total number of $3$-trees of size $w$ in $H^3$ is at most $n(4d^{3})^{w}$.
To see this, first note that the number of 
non-isomorphic trees on $w$ vertices is at most $4^{w}$ (see e.g.~\cite{Lov93}). 
Now fix one such tree and denote it by $T$. 
Label the vertices of $T$ by $v_1, v_2, \ldots, v_w$ in a way 
such that for any $j>1$, vertex $v_j$ is adjacent to some $v_i$ with $i<j$ in $T$. 
How many ways are there to choose $v_1, v_2, \ldots , v_w$ from $V(H)$ so that
they can be the set of vertices in $T$? 
There are $n$ choices for $v_{1}$. 
As $H^3$ has maximum degree $D=d(d-1)^{2}<d^{3}$,
therefore there are at most $D$ possible choices for $v_2$. 
and by induction there are at most 
$nD^{w-1}<n{d^{3w}}$ possible vertex combinations for $T$. 
Since there can be at most $4^{w}$ different $T$'s, it follows that there are at most
$n(4d^3)^w$ possible $3$-trees in $G$.

Since all vertices in $W$ are at least $3$-apart,
all the events $\{B_{v}\}_{v\in W}$ are mutually independent.
Therefore we may upper bound the probability that all vertices in $W$ are surviving vertices as
\begin{align*}
&\quad\Pr[\cap_{v\in W} B_{v}]  \\
&=\prod_{v\in W}\Pr[B_{v}] \\
&\leq\left(\frac{1}{8d^3}\right)^{w},
\end{align*}
where the last inequality follows from Claim~\ref{claim:survived_prob}. 

Now the expected number of $3$-trees of size $w$ is at most
\[
n(4d^{3})^{w}\left(\frac{1}{8d^3}\right)^{w}=n2^{-w}\leq 1/n,
\]
for $w=c_{1}\log{n}$, where $c_{1}$ is some constant. 
By Markov's inequality, with probability at least $1-1/n$,
there is no $3$-tree of size larger than $c_{1}\log{n}$. 
By a simple variant of the $4$-tree
Lemma in~\cite{Bec91} 
(that is, instead of the ``$4$-tree lemma'', we need a ``$3$-tree lemma'' here), 
we see that a connected component of size $s$ in $H$ contains 
a $3$-tree of size at least $s/d^3$. 
Therefore, with probability at least $1-1/n$,
there is no connected surviving vertices of size at least $O(\poly{d} \cdot \log n)$
at the end of Phase $1$ of our algorithm.
\end{proof}

\subsection{Phase 2: Greedy search in the connected component}\label{Sec:max_is_phase2}
If $v$ is a surviving vertex after Phase $1$, 
we need to perform Phase $2$ of the algorithm. 
In this phase, we first explore $v$'s connected component, $C(v)$, in the graph induced 
on $G$ by all the vertices in state ``$\perp$''.
If the size of $C(v)$ is larger than $c_{2}\log n$ for some constant $c_{2}$ 
which depends only on $d$, 
we abort and output ``Fail''.  
Otherwise, we perform the simple greedy algorithm described at the beginning of this section 
to find the MIS in $C(v)$. 
To check if any single vertex is in state ``$\perp$'', 
we simply run our Phase $1$ algorithm on this vertex 
and the running time is $d^{O(d\log d)}$ for each vertex in $C(v)$. 
Therefore the running time for Phase $2$
is at most $O(|C(v)|) \cdot d^{O(d\log d)} \leq O(d^{O(d\log d)}\cdot \log n)$. 
As this dominates the running time of Phase $1$, it is also
the total running time of our local computation algorithm for MIS.

Finally, as we only need to store the random bits generated by each vertex during Phase $1$ of
the algorithm and bookkeep the vertices in the connected component during Phase $2$ 
(which uses at most $O(\log n)$ space), the space complexity of the local computation algorithm
is therefore $O(n)$.

\section{Radio Networks}\label{Sec:radio}
For the purposes of 
this section, a  \emph{radio network} is an undirected graph $G=(V, E)$ with one processor at each vertex.
The processors communicate with each other by transmitting messages in a synchronous
fashion to their neighbors.
In each round, a processor $P$ can either receive a message,
send messages to all of its neighbors, or do nothing.
We will focus on the radio network that is 
referred to as a 
\emph{Type II network}:\footnote{The other model, Type I radio network, is more restrictive: 
A processor $P$ receives a message from its neighbor $Q$ in a given round only if 
$P$ is silent, $Q$ transmits and $P$ chooses to receive from $Q$ in that round.} in \cite{Alo91}.
$P$ receives a message from its neighbor $Q$ if $P$ is silent, 
and $Q$ is the only neighbor of $P$ that transmits in that round. 
Our goal is to check whether there is a two-way connection between 
each pair of adjacent vertices. 
To reach this goal, we would like to find 
a schedule such that each vertex in $G$ broadcasts in one of the $K$ rounds
and $K$ is as small as possible.
\begin{definition}[Broadcast function]
Let $G=(V,E)$ be an undirected graph. We say
$F_{r}: V \to [K]$ is a \emph{broadcast function} for the network $G$
if the following holds:
\begin{enumerate}
\item Every vertex $v$ broadcasts once and only once in round $F_{r}(v)$ 
to all its neighboring vertices;
\item No vertex receives broadcast messages from more than one neighbor in any round;
\item For every edge $(u, v)\in G$, $u$ and $v$ broadcast in 
distinct rounds.
\end{enumerate}
\end{definition}

Let $\Delta$ be the maximum degree of $G$.
Alon et. al.~\cite{ABLP89, ABLP92} show that 
the minimum number of rounds $K$ satisfies $K=\Theta(\Delta \log \Delta)$.
Furthermore, Alon~\cite{Alo91} gives an $NC_{1}$ algorithm that computes 
the broadcast function $F_{r}$ with $K=O(\Delta \log \Delta)$.
Here we give a local computation algorithm for this problem, 
i.e. given a degree-bounded graph $G=(V,E)$ in the adjacency list form and a vertex $v \in V$, we
output the round number in which $v$ broadcasts in logarithmic time. 
Our solution is \emph{consistent} in the sense that 
all answers our algorithm outputs to the various $v \in V$
agree with some broadcast scheduling function $F_{r}$. 

Let $G^{1,2}$ be the ``square graph'' of $G$; that is, $u$ and $v$ are
connected in $G^{1,2}$
if and only if their distance in $G$ is either one or two.
Our algorithm is based on finding an \emph{independent set cover} of $G^{1,2}$
which simulates Luby's Maximal Independent Set algorithm~\cite{Lub86}.
Note that if we denote the maximum degree of $G^{1,2}$ by $d$, then $d\leq \Delta^{2}$.

\begin{definition}[Independent Set Cover]
Let $H=(V,E)$ be an undirected graph.
A collection of vertex subsets 
$\{S_{1}, \ldots, S_{t}\}$ is an \emph{independent set cover} (ISC) for 
$H$ if these vertex sets are pairwise disjoint, each $S_{i}$ is an independent set
in $H$ and their union equals $V$. We call $t$ the \emph{size} of ISC $\{S_{1}, \ldots, S_{t}\}$.
\end{definition}

\begin{fact}\label{fact:ISC}
If $\{S_{1}, \ldots, S_{t} \}$ is an ISC for 
$G^{1,2}$, then the function defined by $F_r(v)=i$ iff $v\in S_{i}$ is a broadcast function.
\end{fact}
\begin{proof}
First note that, since the union of $\{S_{i}\}$ equals $V$, $F_{r}(v)$ is well-defined for 
every $v\in G$. That is, every $v$ broadcasts 
in some round in $[t]$, hence both directions of
every edge are covered in some round. 
As $v$ can only be in one $IS$, it only broadcasts once.
Second, for any two vertices $u$ and $v$, if $d(u, v)\geq 3$, 
then $N(u)\cap N(v)=\emptyset$. It follows that, if in each round all the 
vertices that broadcast are at least $3$-apart from each other, no vertex will receive
more than one message in any round. Clearly the vertices in an independent set of $G^{1,2}$
have the property that all the pairwise distances are at least $3$.
\end{proof}

The following is a simple fact about ISCs.
\begin{fact}\label{fact:ISC_alg}
For every undirected graph $H$ on $n$ vertices with maximum degree $d$, 
there is an ISC of size at most $d$. Moreover, such an ISC can be found 
by a greedy algorithm in time at most $O(dn)$.
\end{fact}
\begin{proof}
We repeatedly apply the greedy algorithm that finds an MIS in order to find an ISC. 
Recall that the greedy algorithm repeats the following until the graph has no unmarked vertex:
pick an unmarked vertex $v$, add it to the IS and mark off all the vertices in $N(v)$.
Clearly each IS found by the greedy algorithm has size at least $\frac{n}{d+1}$.
To partition the vertex set into an ISC,
we run this greedy algorithm to find an IS which we call $S_{1}$, and delete all 
the vertices in $S_{1}$ from the graph. 
Then we run the greedy algorithm on the new graph again to get $S_2$, and so on.
After running at most $d$ rounds (since each round reduces the maximum degree of the graph by at least one), 
we partition all the vertices into 
an ISC of size at most $d$ and the total running time is at most $O(dn)$.
\end{proof}

Our main result in this section is a local computation algorithm that computes an 
ISC of size $O(d\log d)$ for any graph of maximum degree $d$. On input a vertex $v$, 
our algorithm outputs the index $i$ of a vertex subset $S_{i}$ to which $v$ belongs, in an ISC of $H$. 
We will call $i$ the \emph{round number} of $v$ in the ISC.
By Fact~\ref{fact:ISC}, applying this algorithm to graph $G^{1,2}$
gives a local computation algorithm that computes a broadcast function for $G$.

\subsection{A local computation algorithm for ISC}
Our main result for computing an ISC is summarized in the following theorem.
\begin{theorem}\label{thm:ISC_main}
Let $H$ be an undirected graph on $n$ vertices with maximum degree $d$.
Then there is a $(\poly{d}\cdot \log{n}, O(n), 1/n)$-local 
computation algorithm 
which, on input a vertex $v$, computes the round number of $v$
in an ISC of size at most $O(d\log d)$.
Moreover, the algorithm will give a consistent ISC for every vertex in $H$.
\end{theorem}

On input a vertex $v$, 
our algorithm computes the round number of $v$ in two phases. 
In Phase 1 we simulate Luby's algorithm for MIS~\cite{Lub86}
for $O(d\log d)$ rounds.
At each round, $v$ tries to put itself in the independent set generated in that round.
That is, $v$ chooses itself with probability $1/2d$ and if none of its 
neighbors choose themselves, then $v$ is selected in that round and we output that round number for $v$. 
As we show shortly, after Phase 1,
most vertices will be assigned a round number.   We say $v$ {\em survives}
if it is not assigned a round number. We consider
the connected component containing $v$ after one deletes all vertices
that do not survive from the graph.
Following an argument similar to that of Beck~\cite{Bec91}, almost surely,
all such connected components of surviving vertices after Phase 1
have size at most $O(\log n)$. This enables us, in Phase 2, to
perform the greedy algorithm on $v$'s connected component
to deterministically compute the round number of $v$ in time
$O(\log n)$.

\subsubsection{Phase 1 algorithm}
Phase 1 of our local computation algorithm for computing an ISC is shown in Figure~\ref{starf1}.
\footnote{In 2(b), we flip random coins for $u$ even if $u$ is selected in
a previous round.  We do this for the 
technical reason that we want to 
rid the dependency of $v$ on nodes that are not neighbors
to simplify our analysis.
Thus our analysis is overly pessimistic since
if selected neighbors stop choosing themselves,
it only increases the chance of $v$ being selected.}

\begin{figure*}
\begin{center}
\fbox{
\begin{minipage}{5in}
\small
\begin{tabbing}
\textsc{Independent Set Cover: Phase 1} \\
Input: a graph $H$ and a vertex $v \in V$\\
Out\=put\=:   the \=rou\=nd \=num\=ber of $v$ in the ISC or ``$\perp$''\\
1. \> Initialize all vertices in $N^{+}(v)$ to state ``$\perp$'' \\
2. \> For $i=1$ to $r=20d\log{d}$ \\
   \>\> (a) If $v$ is labeled "$\perp$"\\
   \>\>\>  $v$ chooses itself independently with probability $\frac{1}{2d}$\\
   \>\> (b) If $v$ chooses itself\\
   \>\>\>  (i) For every $u \in N(v)$\\
   \>\>\>\>  (even if $u$ is labeled ``selected in round $j$''\\
   \>\>\>\>    for some $j<i$, we still flip random coins for it)\\
   \>\>\>\>  $u$ chooses itself independently with probability $\frac{1}{2d}$\\
   \>\>\>  (ii) If $v$ has a chosen neighbor, \\
   \>\>\>\>   $v$ unchooses itself\\
   \>\>\>  (iii) Else\\
   \>\>\>\>  $v$ is labeled ``selected in round $i$''\\
   \>\>\>\>  return $i$ \\
3. \>    return ``$\perp$''
\end{tabbing}
\end{minipage}
}
\end{center}
\caption{Algorithm for finding an Independent Set Cover: Phase 1.}
\label{starf1}
\end{figure*}

For every $v\in V$, let $A_{v}$ be the event that vertex $v$ returns ``$\perp$'', i.e.  
$v$ is not selected after $r$ rounds. We call such a $v$ 
a \emph{surviving} vertex.   After deleting all $v$ that do not
survive from the graph, we are interested in bounding the size
of the largest remaining connected component.
Clearly event $A_v$ depends on the random coin tosses of $v$ and $v$'s neighboring 
vertices in all the $r$ rounds. A graph $H$ on the vertices $V(H)$ (the indices for the $A_{v}$)
is called a \emph{dependency graph} for $\{A_{v}\}_{v\in V(H)}$ if for all $v$ the
event $A_{v}$ is mutually independent of all $A_{u}$ with $(u,v)\notin H$.

The following two claims are identical to 
Claim~\ref{claim:survived_degree} and Claim~\ref{claim:survived_prob} 
in Section~\ref{Sec:maxis} respectively, we therefore omit the proofs. 
\begin{claim}\label{claim:survived_degree_radio}
The dependency graph $H$ has maximum degree $d^2$.
\end{claim} 

\begin{claim}\label{claim:survived_prob_radio}
For every $v\in V$, the probability that $A_{v}$ occurs is at most $1/8d^3$.
\end{claim}

The following observation is crucial in our local computation algorithm.
\begin{lemma}\label{lemma:survived_size_radio}
After Phase 1, with probability at least $1-1/n$, all connected components 
of the surviving vertices are of size at most $O(\poly{d} \cdot \log n)$. 
\end{lemma}
\begin{proof}
The proof is almost identical to that of Lemma~\ref{lemma:survived_size} but
is only simpler: we can directly upper bound the probability
\[
\Pr[\text{all vertices in $W$ are surviving vertices}]
=\Pr[\cap_{v\in W} A_{v}] 
\]
by way of Beck~\cite{Bec91} without resorting to any other random process. 
We omit the proof.
\end{proof}

\subsubsection{Phase 2 algorithm}
If $v$ is a surviving vertex after Phase 1, 
we perform Phase 2 of the algorithm. 
In this phase, we first explore the connected component, $C(v)$,
that the surviving
vertex $v$ lies in.
If the size of $C(v)$ is larger than $c_{2}\log n$ for some constant $c_{2}(d)$ depending only on $d$, 
we abort and output ``Fail''.  Otherwise,
we perform the simple greedy algorithm described in Fact~\ref{fact:ISC_alg} to 
partition $C(v)$ into at most $d$ subsets deterministically. The running time for Phase 2
is at most $\poly{d}\cdot \log n$. Since any independent set of a connected component
can be combined with independent sets of other connected components to form an IS for 
the surviving vertices, we conclude that the total size of ISC we find is 
$O(d\log d)+d=O(d\log d)$.

\subsection{Discussions}
Now a simple application of Theorem~\ref{thm:ISC_main} to $G^{1,2}$ 
gives a local computation algorithm for the broadcast function.

\begin{theorem}\label{sflog}
Given a graph $G=(V,E)$ with $n$ vertices and maximum degree $\Delta$ and a vertex $v \in V$, 
there exists a 
$(\poly{\Delta}\cdot \log n, O(n), 1/n)$-local computation algorithm
that computes a broadcast function with at most $O(\Delta^{2}\log \Delta)$ rounds. 
Furthermore, the broadcast function it outputs is consistent 
for all queries to the vertices of the graph.
\end{theorem}

We note our round number bound is quadratically larger than that of Alon's parallel algorithm~\cite{Alo91}.  
We do not know how to turn his algorithm into a local computation algorithm.

\section{Hypergraph two-coloring}\label{Sec:hypergraph}

A \emph{hypergraph} $H$ is a pair $H = (V,E)$ where $V$ is a finite set whose elements are
called \emph{nodes} or \emph{vertices}, and $E$ is a family of non-empty subsets of $V$,
called \emph{hyperedges}. A hypergraph is called \emph{$k$-uniform} if each of its
hyperedges contains precisely $k$ vertices.
A \emph{two-coloring} of a hypergraph $H$ is a mapping $\mathbf{c}: V\to \{\text{red, blue}\}$
such that no hyperedge in $E$ is monochromatic.
If such a coloring exists, then we say $H$ is \emph{two-colorable}.
We assume that each
hyperedge in $H$ intersects at most $d$ other hyperedges.
Let $N$ be the number of hyperedges in $H$. Here we think of $k$ and $d$ as fixed constants
and all asymptotic forms are with respect to $N$.
By the Lov{\'{a}}sz Local Lemma,
when $e(d+1) \leq 2^{k-1}$, the hypergraph $H$ is
two-colorable.

Let $m$ be the total number of vertices in $H$. Note that $m\leq kN$, so $m=O(N)$.
For any vertex $x\in V$, we use $\mathcal{E}(x)$ to denote the set of hyperedges $x$ belongs to.
For convenience, for any hypergraph $H = (V,E)$,
we define an $m$-by-$N$ \emph{vertex-hyperedge incidence matrix} $\mathcal{M}$
such that, for any vertex $x$ and hyperedge $e$, $\mathcal{M}_{x,e}=1$ if $e\in \mathcal{E}(x)$ and 
$\mathcal{M}_{x,e}=0$ otherwise.
A natural representation of the input hypergraph $H$ is this vertex-hyperedge incidence matrix $\mathcal{M}$.
Moreover, since we assume both $k$ and $d$ are constants, 
the incidence matrix $\mathcal{M}$ is necessarily very sparse. 
Therefore, we further assume that the matrix $\mathcal{M}$ is implemented via
linked lists for each row (that is, vertex $x$) and each column (that is, hyperedge $e$). 

Let $G$ be the \emph{dependency graph} of the hyperedges in $H$. 
That is, the vertices of the undirected graph $G$
are the $N$ hyperedges of $H$ and a hyperedge $E_{i}$ is connected to
another hyperedge $E_{j}$ in $G$ if $E_{i}\cap E_{j} \neq \emptyset$.
It is easy to see that if the input hypergraph is given in the 
above described representation, then we can find all the neighbors of any hyperedge $E_{i}$ 
in the dependency graph $G$ (there are at most $d$ of them) in $O(\log N)$ time.

\subsection{Our main result}
A natural question to ask is:
Given a two-colorable hypergraph $H$ and a vertex $v\in V(H)$,
can we quickly compute the coloring of $v$? 
Here we would like the coloring to be \emph{consistent}, 
meaning all the answers we provide must come from the
{\em same} valid two-coloring.
Our main result in this section is,
given a two-colorable hypergraph $H$ whose two-coloring scheme is guaranteed by
the Lov{\'{a}}sz Local Lemma (with slightly weaker parameters),
we give a local computation algorithm which answers queries of the coloring
of any single vertex in $\polylog{N}$ time, 
where $N$ is the number of the hyperedges in $H$.
The coloring returned by our oracle 
will agree with some two-coloring of the hypergraph with probability at least $1-1/N$.

\begin{theorem}\label{thm:hypergraph_main}
Let $d$ and $k$ be such that there exist three positive integers $k_{1}, k_{2}$ and $k_{3}$ such that
the followings hold:
\begin{align*}
k_{1}+k_{2}+{k_3} &=k, \\
16d(d-1)^{3}(d+1) &< 2^{k_{1}},\\
16d(d-1)^{3}(d+1) &< 2^{k_{2}},\\
2e(d+1) &< 2^{k_{3}}.
\end{align*}
Then there exists a  
$(\polylog{N}, O(N), 1/N)$-local computation algorithm 
which, given a hypergraph $H$ and any sequence of
queries to the colors of vertices $(x_1, x_2, \ldots, x_s)$, 
returns a consistent coloring for all $x_i$'s which 
agrees with some $2$-coloring of $H$.  
\end{theorem}

\subsection{Overview of the coloring algorithm}\label{Sec:hypergraph_overview}
Our local computation algorithm imitates the parallel coloring algorithm of Alon~\cite{Alo91}.
Recall that Alon's algorithm runs in three phases. 
In the first phase, we randomly color each vertex in
the hypergraph following some arbitrary ordering of the vertices. 
If some hyperedge has $k_{1}$ vertices in one color and 
no vertices in the other color,
we call it a \emph{dangerous} edge and 
mark all the remaining vertices in that hyperedge as \emph{troubled}.
These \emph{troubled} vertices will not be colored in the first phase.
If the queried vertex becomes a \emph{troubled} vertex from the coloring process
of some previously queried vertex,
then we run the Phase $2$ coloring algorithm.   
There we first delete all hyperedges which have been assigned both colors and call the remaining
hyperedges \emph{surviving} edges. 
Then we repeat the same process again for the \emph{surviving} hyperedges, but this time 
a hyperedge becomes dangerous if $k_{1}+k_{2}$ vertices are colored the same color and
no vertices are colored by the other color. 
Finally, in the third phase, we do a brute-force search for a coloring
in each of the connected components of the surviving vertices 
as they are of size $O(\log\log N)$ almost surely.

A useful observation is, in the first phase of Alon's algorithm, 
we can color the vertices in \emph{arbitrary} order. 
In particular, this order can be taken to be the 
order that queries to the local computation algorithm are made in. 
If the coloring of a vertex $x$ can not be determined in the first phase, 
then we explore the dependency graph around the hyperedges containing $x$ 
and find the connected component of the \emph{surviving} hyperedges to perform the second phase coloring. 
To ensure that \emph{all} the connected components of \emph{surviving hyperedges} 
resulting from the second phase coloring are of small sizes,
we repeat the second phase colorings independently many times until
the connected components sizes are small enough.
If that still can not decide the coloring of $x$, then we
run the third (and final) phase of coloring, in which we exhaustively search for a two-coloring
for vertices in some very small (i.e., of size at most $O(\log\log N)$) 
connected component in $G$ as guaranteed by our second phase coloring.
Following Alon's analysis, we show that with probability at least $1-1/N$,
the total running time of all these three phases for any vertex in $H$ is $\polylog{N}$.

During the execution of the algorithm, each hyperedge will be in
either \emph{initial}, \emph{safe}, \emph{unsafe-$1$}, \emph{unsafe-$2$},
\emph{dangerous-$1$} or \emph{dangerous-$2$} state.
Vertices will be in either \emph{uncolored}, \emph{red}, \emph{blue}, 
\emph{trouble-$1$} or \emph{trouble-$2$} state.
The meanings of all these states should be clear
from their names.
Initially every hyperedge is in \emph{initial} state and every vertex is in \emph{uncolored} state.

\subsection{Phase 1 coloring}\label{Sec:hypergraph_P1}

\begin{figure*}
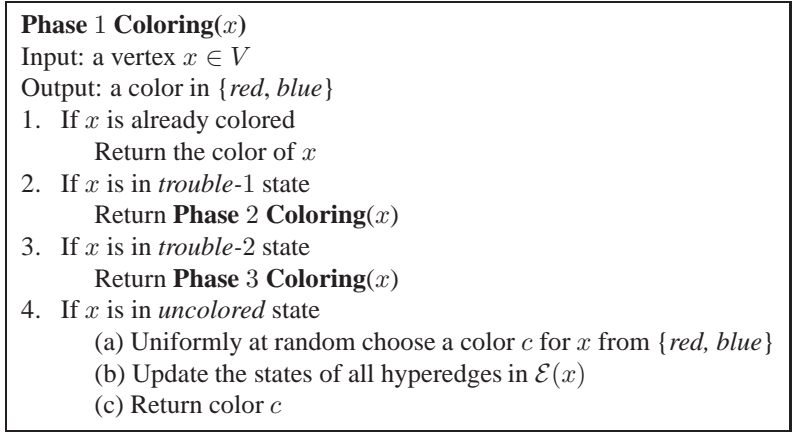

\begin{center}
\fbox{
\begin{minipage}{5in}
\small
\begin{tabbing}
\textbf{Phase $1$ Coloring($x$)} \\
Input: a vertex $x \in V$\\
Out\=put\=: a color in \{\emph{red}, \emph{blue}\}\\
1.\> If $x$ is already colored \\
  \>\> Return the color of $x$ \\
2.\> If $x$ is in \emph{trouble-$1$} state\\
  \>\> Return \textbf{Phase $2$ Coloring}($x$)\\
3.\> If $x$ is in \emph{trouble-$2$} state\\
  \>\> Return \textbf{Phase $3$ Coloring}($x$)\\
4.\> If $x$ is  in \emph{uncolored} state \\
  \>\> (a) Uniformly at random choose a color $c$ for $x$ from \{\emph{red, blue}\} \\
  \>\> (b) Update the states of all hyperedges in $\mathcal{E}(x)$ \\
  \>\> (c) Return color $c$
\end{tabbing}
\end{minipage}
}
\end{center}
\caption{Phase $1$ coloring algorithm}
\label{Fig:Phase1}
\end{figure*}

If $x$ is already colored (that is,
$x$ is in either \emph{red} or \emph{blue} state), then we simply return that color.
If $x$ is in the \emph{trouble-$1$} state, we invoke Phase $2$ coloring for vertex $x$.
If $x$ is in the \emph{trouble-$2$} state, we invoke Phase $3$ coloring for vertex $x$.
If $x$ is \emph{uncolored}, then we flip a fair coin to color $x$
red or blue with equal probability 
(that is, vertex $x$'s state becomes \emph{red} or \emph{blue}, respectively). 
After that, we update the status of all the
hyperedges in $\mathcal{E}(x)$. Specifically, if some $E_{i} \in \mathcal{E}(x)$
has $k_{1}$ vertices in one color and no vertices in the other color, then
we change $E_{i}$ from \emph{initial} into \emph{dangerous-$1$} state.
Furthermore, all uncolored vertices in $E_{i}$ 
will be changed to \emph{trouble-$1$} states.
On the other hand, if both colors appear among the
vertices of $E_{i}$, we update the
state of $E_{i}$ from \emph{initial} to \emph{safe}.
If none of the vertices in a hyperedge is \emph{uncolored} and the hyperedge is still
in \emph{initial} state (that is, it is neither \emph{safe} or \emph{dangerous-$1$}),
then we change its state to \emph{unsafe-$1$}.
Note that if a hyperedge is \emph{unsafe-$1$} then all of its
vertices are either colored or in \emph{trouble-$1$} state, 
and the colored vertices are monochromatic.

\paragraph{Running time analysis.} 
The running time of Phase $1$ coloring for an \emph{uncolored} vertex 
$x$ is $O(kd)=O(1)$ (recall that we assume both $k$ and $d$ are constants). 
This is because vertex $x$ can belong to at most $d+1$ hyperedges, hence there are at most
$k(d+1)$ vertices that need to be updated during Phase $1$. 
If $x$ is already a \emph{colored} vertex, 
the running time is clearly $O(1)$. 
Finally, the running time of Phase $1$ coloring 
for a \emph{trouble-$1$} or \emph{trouble-$2$}
vertex is $O(1)$ plus the running time of Phase $2$ coloring 
or $O(1)$ plus the running time of Phase $3$ coloring, respectively.

\subsection{Phase 2 coloring}\label{Sec:hypergraph_P2}
\begin{figure*}
\begin{center}
\fbox{
\begin{minipage}{8in}
\small
\begin{tabbing}
\textbf{Phase $2$ Coloring($x$)} \\
Input: a \emph{trouble-$1$} vertex $x \in V$ \\
Out\=put\=: a \=color in \{\emph{red}, \emph{blue}\} or \emph{FAIL}\\
1.\> Start from $\mathcal{E}(x)$ to explore $G$ in order to find the connected \\
  \>\>components of all the \emph{surviving-$1$} hyperedges around $x$\\
2.\> If the size of the component is larger than $c_{1}\log N$ \\
  \>\> Abort and return \emph{FAIL}\\
3.\> Repeat the following $O(\frac{\log{N}}{\log\log{N}})$ times and stop if
  a \emph{good} coloring is found\\
   \>\> (a) Color all the vertices in $C_{1}(x)$ uniformly at random\\
   \>\> (b) Explore the dependency graph of $G|_{S_{1}(x)}$ \\
   \>\> (c) Check if the coloring is \emph{good} \\
4.\> Return the color of $x$ in the good coloring
\end{tabbing}
\end{minipage}
}
\end{center}
\caption{Phase $2$ coloring algorithm}
\label{Fig:Phase2}
\end{figure*}

During the second phase of coloring, 
given an input vertex $x$ (which is necessarily a \emph{trouble-$1$}),
we first explore the dependency graph $G$ of the hypergraph $H$ 
by keep coloring some other vertices whose colors
may have some correlation with the coloring of $x$.
In doing so, we grow a connected component of \emph{surviving-$1$} hyperedges containing $x$ in $G$. 
Here, a hyperedge is called \emph{surviving-$1$} 
if it is either \emph{dangerous-$1$} or \emph{unsafe-$1$}.
We denote this connected component of \emph{surviving-$1$} hyperedges 
surrounding vertex $x$ by $C_{1}(x)$.

\paragraph{Growing the connected component.}
Specifically, in order to find out $C_{1}(x)$, we maintain a set of hyperedges $\mathcal{E}_{1}$ 
and a set of vertices $V_{1}$. Throughout the process of exploring $G$, 
$V_{1}$ is the set of \emph{uncolored} vertices that are contained in some hyperedge
in $\mathcal{E}_{1}$. Initially $\mathcal{E}_{1}=\mathcal{E}(x)$.
Then we independently color each vertex in $V_{1}$ \emph{red} or \emph{blue} uniformly at random.
After coloring each vertex, we update the state of every hyperedge that contains the vertex. 
That is,
if any hyperedge $E_{i} \in V_{1}$ becomes \emph{safe}, then we remove
$E_{i}$ from $V_{1}$ and delete all the vertices that are \emph{only} contained in $E_{i}$.
On the other hand, once a hyperedge in $V_{1}$ becomes \emph{dangerous-$2$} 
(it has $k_2$ vertices,
all the \emph{uncolored} vertices in that hyperedge become \emph{trouble-$2$} and
we skip the coloring of all such vertices.
After the coloring of all vertices in $V_{1}$, hyperedges in $\mathcal{E}_{1}$ are surviving hyperedges.
Then we check all the hyperedges in $G$ that are adjacent to the hyperedges in $\mathcal{E}_{1}$.
If any of these hyperedges is not in the \emph{safe} state, then we add it to $\mathcal{E}_{1}$ and also
add all its \emph{uncolored} vertices to $V_{1}$.
Now we repeat the coloring process described above for these newly added \emph{uncolored} vertices.
This exploration of the dependency graph terminates if, either there is no more hyperedge to color,
 or the number of \emph{surviving-$1$} hyperedges in $\mathcal{E}_{1}$ is greater than $c_{1}\log N$,
 where $c_{1}$ is some absolute constant.
The following Lemma shows that, almost surely, the size of $C_{1}(x)$ is at most $c_{1}\log N$.

\begin{lemma}[\cite{Alo91}]\label{lemma:P2-Alon}
Let $S\subseteq G$ be the set of surviving hyperedges after the first phase. Then with probability
at least $1-\frac{1}{2N}$ 
(over the choices of random coloring), all connected components $C_{1}(x)$ of $G|_{S}$
have sizes at most $c_{1}\log N$.
\end{lemma}

\paragraph{Random coloring.}
Since $C_{1}(x)$ is not connected to any \emph{surviving-$1$} hyperedges in $H$, 
we can color the vertices in the connected component $C_{1}(x)$ 
without considering any other hyperedges that are outside $C_{1}(x)$.
Now we follow a similar coloring process as in Phase $1$ to color the vertices in $C_{1}(x)$
uniformly at random and in an arbitrary ordering.
The only difference is, 
we ignore all the vertices that are already colored \emph{red} or \emph{blue},
and if $k_{1}+k_{2}$ vertices in a hyperedge get colored monochromatically,
and all the rest of vertices in the hyperedge are in \emph{trouble-$1$} state,
then this hyperedge will be in \emph{dangerous-$2$} state and all the uncolored vertices in it 
will be in \emph{trouble-$2$} state. 
Analogously we define \emph{unsafe-$2$} hyperedges as 
hyperedges whose vertices are either colored or in \emph{trouble-$2$} state and 
all the colored vertices are monochromatic.
Finally, we say a hyperedge is a \emph{surviving-$2$} edge if it is in
either \emph{dangerous-$2$} state or \emph{unsafe-$2$} state.

Let $S_{1}(x)$ be the set of surviving hyperedges in $C_{1}(x)$
after all vertices in $C_{1}(x)$ are either colored or in \emph{trouble-$2$} state.
Now we explore the dependency graph of $S_{1}(x)$ to find out all
the connected components. Another application of Lemma~\ref{lemma:P2-Alon}
to $G|_{S_{1}(x)}$ shows that with probability 
at least $1-O(\frac{1}{\log^{2}N})$ (over the choices of random coloring), 
all connected components in $G|_{S_{1}(x)}$
have sizes at most $c_{2}\log\log N$, where $c_{2}$ is some constant.
We say a Phase $2$ coloring is \emph{good} if this condition is satisfied.
Now if a random coloring is not good, then 
we erase all the coloring performed during Phase $2$ and repeat the above coloring
and exploring dependency graph process.
We keep doing this until we find a
good coloring. Therefore, 
after recoloring at most $O(\frac{\log{N}}{\log\log{N}})$ times (and therefore with
at most $\polylog{N}$ running time),
we can, with probability at least $1-1/2N^2$, color $C_{1}(x)$ such that 
each connected component in $G|_{S_{1}(x)}$ has size at most
$c_{2}\log\log N$.
By the union bound, with probability at least $1-1/2N$,
the Phase $2$ colorings for all connected components find some good colorings.

\paragraph{Running time analysis.}
Combining the analysis above with an argument similar to the running time analysis
of Phase $1$ coloring gives
\begin{claim}\label{claim:P2-time}
Phase $2$ coloring takes at most $\polylog{N}$ time.
\end{claim}

\subsection{Phase 3 coloring}\label{Sec:hypergraph_P3}

\begin{figure*}
\begin{center}
\fbox{
\begin{minipage}{8in}
\small
\begin{tabbing}
\textbf{Phase $3$ Coloring($x$)} \\
Input: a \emph{trouble-$2$} vertex $x \in V$ \\
Out\=put\=: a \=color in \{\emph{red}, \emph{blue}\}\\
1.\> Start from $\mathcal{E}(x)$ to explore $G$ in order to find the connected \\
  \>\> component of all the \emph{surviving-$2$} hyperedges around $x$\\
2.\> Go over all possible colorings of the connected component\\
  \>\> and color it using a feasible coloring. \\
3.\> Return the color $c$ of $x$ in this coloring.
\end{tabbing}
\end{minipage}
}
\end{center}
\caption{Phase $3$ coloring algorithm}
\label{Fig:Phase3}
\end{figure*}

In Phase $3$, given a vertex $x$ (which is necessarily \emph{trouble-$2$}), 
we grow a connected component which includes $x$ as in Phase $2$, 
but of \emph{surviving-$2$} hyperedges. 
Denote this connected component of \emph{surviving-$2$} hyperedges by $C_{2}(x)$.
By our Phase $2$ coloring, the size of $C_{2}(x)$ is
no greater than $c_{2} \log\log N$. 
We then color the vertices in this connected component by exhaustive search.
The existence of such a coloring is guaranteed by the 
 Lov{\'{a}}sz Local Lemma (Lemma ~\ref{lemma:LLL}).
 
\begin{claim}\label{claim:P3-time}
The time complexity of Phase $3$ coloring is at most $\polylog{N}$. 
\end{claim}
\begin{proof}
Using the same analysis as for Phase $2$, 
in time $O(\log\log N)$ we can explore the dependency graph
to grow our connected component of \emph{surviving-$2$} hyperedges. 
Exhaustive search of a valid two-coloring of all the vertices in $C_{2}(x)$ takes time at most 
$2^{O(|C_{2}(x)|)}=2^{O(\log\log N)}=\polylog{N}$, as $|C_{2}(x)| \leq c_2 \log \log N$ and
each hyperedge contains $k$ vertices.
\end{proof}

Finally, we remark that using the same techniques as those in~\cite{Alo91}, 
we can make our local computation algorithm run in parallel and 
find an $\ell$-coloring of a hypergraph for any $\ell\geq 2$
(an $\ell$-coloring of a hypergraph is to color each vertex 
in one of the $\ell$ colors such that each color appears in every hyperedge).

\section{\texorpdfstring{$k$-CNF}
{k-CNF}}\label{Sec:kcnf}

As another example, we show our hypergraph coloring algorithm can be easily modified to 
compute a satisfying assignment of a $k$-CNF formula, 
provided that the latter satisfies some specific properties.

Let $H$ be a $k$-CNF formula on $m$ Boolean variables $x_{1}, \ldots, x_{m}$.
Suppose $H$ has $N$ clauses $H=A_{1} \wedge \cdots \wedge A_{N}$
and each clause consists of exactly $k$ distinct literals.\footnote{
Our algorithm works for the case that each clause has at least $k$ literals; 
for simplicity, we assume that all clauses have uniform size.}
We say two clauses $A_{i}$ and $A_{j}$ \emph{intersect} 
with each other if they
share some variable (or the negation of that variable).
As in the case for hypergraph coloring, $k$ and $d$ are fixed constants
and all asymptotics are with respect to the number of clauses $N$ (and hence $m$, since $m\leq kN$).
Our main result is the following.

\begin{theorem}
Let $H$ be a $k$-CNF formula with $k\geq 2$. If each clause intersects no more
than $d$ other clauses and furthermore $k$ and $d$ 
are such
that there exist three positive integers $k_{1}, k_{2}$ and $k_{3}$ satisfying
the followings relations:
\begin{align*}
k_{1}+k_{2}+{k_3} &=k, \\
8d(d-1)^{3}(d+1) &< 2^{k_{1}},\\
8d(d-1)^{3}(d+1) &< 2^{k_{2}},\\
e(d+1) &< 2^{k_{3}},
\end{align*}
then there exists a local computation algorithm that, given any sequence of
queries to the truth assignments of variables $(x_1, x_2, \ldots, x_s)$, 
with probability at least $1-1/N$,
returns a consistent truth assignment for all $x_i$'s which agrees with some 
satisfying assignment of the $k$-CNF formula $H$. 
Moreover, the algorithm answers each single query in
$O((\log N)^c)$ time, where $c$ is some constant (depending only on $k$ and $d$). 
\end{theorem}

\begin{sketch}
We follow a similar
algorithm to that of 
hypergraph two-coloring as presented in Section~\ref{Sec:hypergraph}.
Every clause will be in
either \emph{initial}, \emph{safe}, \emph{unsafe-$1$}, \emph{unsafe-$2$},
\emph{dangerous-$1$} or \emph{dangerous-$2$} state.
Every variable will be in either \emph{unassigned}, \emph{true-$1$}, \emph{false-$1$}, 
\emph{trouble-$1$} or \emph{trouble-$2$} state.
Initially every clause is in \emph{initial} state and every variable is in \emph{unassigned} state.
Suppose we are asked about the value of a variable $x_{i}$. 
If $x_{i}$ is in \emph{initial} state,
we randomly choose from $\{\true, \false\}$ 
with equal probabilities and assign it to $x_{i}$. 
Then we update all the clauses that contain either $x_{i}$ or $\bar{x}_{i}$ accordingly: 
If the clause is already evaluated to
$\true$ by this assignment of $x_{i}$, then we mark the literal as \emph{safe};
if the clause is in \emph{initial} state and is not \emph{safe} yet and 
$x_{i}$ is the $k^{\text{th}}_{1}$ literal in the clause that has been assigned values,
then the clause is marked as \emph{dangerous-$1$} and all the remaining unassigned variables in that clause
are now in \emph{trouble-$1$} state. 
We perform similar operations for clauses in other states as we do for the
hypergraph coloring algorithm.
The only difference is now we have $\Pr[\text{$A_{i}$ becomes \emph{dangerous-$1$}}]=2^{-k_{1}}$,
instead of $2^{1-k_{1}}$ as in the hypergraph coloring case.
Following the same analysis, almost surely, all connected components in the dependency graph 
of \emph{unsafe-$1$} clauses are of size at most $O(\log N)$ and
almost surely all connected components in the dependency graph 
of \emph{unsafe-$2$} clauses are of size at most $O(\log\log N)$,
which enables us to do exhaustive search to find a satisfying assignment.
\end{sketch}

\section{Concluding Remarks and Open Problems}\label{Sec:conclusion}
In this paper we propose a model of local computation algorithms
and give some techniques which can be applied to construct local computation algorithms
with polylogarithmic time and space complexities.
It would be interesting to understand the scope of problems which
can be solved with such algorithms and to develop other techniques
that would apply in this setting.

\section*{Acknowledgments}
We would like to thank Ran Canetti,
Tali Kaufman, Krzysztof Onak 
and the anonymous referees for useful discussions and suggestions.
We thank Johannes Schneider for his help with the references.

\bibliographystyle{plain}

\end{document}